\documentclass[ 11pt,reqno ]{amsart}

\AddToHook{bfseries}{\boldmath} 

\usepackage[letterpaper,twoside=false,scale=0.8 ]{geometry}

\usepackage[ english ]{babel}
\usepackage{microtype}

\usepackage{graphicx}
\usepackage[dvipsnames]{xcolor}

\usepackage{amsthm}
\usepackage{mathtools}

\usepackage[shortlabels,inline]{enumitem}
\usepackage[T1]{fontenc}
\usepackage{libertinus}
\usepackage{libertinust1math}
\usepackage[ scr=stixtwofancy,bb=stixtwo,cal=stixtwoplain ]{mathalpha}

\usepackage{subcaption}

\usepackage[ bookmarksopen,bookmarksdepth=3,colorlinks=true,urlcolor=NavyBlue,citecolor=ForestGreen,linkcolor=BrickRed,hypertexnames=false ]{hyperref}

\theoremstyle{plain}
\newtheorem{theorem}{Theorem}[section]

\newtheorem{corollary}[theorem]{Corollary}
\newtheorem{lemma}[theorem]{Lemma}
\newtheorem*{mainresult*}{Main Result}

\theoremstyle{remark}
\newtheorem{remark}[theorem]{Remark}

\theoremstyle{definition}

\usepackage{csquotes} 
\usepackage[ bibstyle=ieee,citestyle=numeric-comp,dashed=true,sorting=nty,backref=true,maxnames=7 ]{biblatex} 
\DeclareNumChars*{+} 

\usepackage[ capitalize,noabbrev,nameinlink ]{cleveref}
\crefname{mainresult*}{Main Result}{Main Results}

\usepackage{pgf}
\usepackage{tikzexternal}

\usepackage{tikzscale}
\tikzsetexternalprefix{tikz/}
\DeclareMathOperator{\res}{res}
\DeclareMathOperator{\sgn}{sgn}

\DeclarePairedDelimiter{\parn}{\lparen}{\rparen}
\DeclarePairedDelimiter{\brak}{\lbrack}{\rbrack}
\DeclarePairedDelimiter{\set}{\lbrace}{\rbrace}

\DeclarePairedDelimiter{\case}{\lbrace}{.}

\newcommand{\mb}{\mathbb} 

\newcommand{\ceq}{\coloneq} 
\newcommand{\eqc}{\eqcolon} 
\newcommand{\mo}{\mathring}

\newcommand{\R}{\mb{R}}

\numberwithin{equation}{section}

\title[Two-layer conjugate flows in an open channel]
{Two-layer conjugate flows in an open channel}

\author[Ming Chen]{Robin Ming Chen}
\address{Department of Mathematics, University of Pittsburgh, PA 15260.}
\email{mingchen@pitt.edu}

\author[Kristoffer Varholm]{Kristoffer Varholm}
\address{Department of Mathematics, University of Pittsburgh, PA 15260.}
\email{kristoffer.varholm@pitt.edu}

\author[Annie G. Wang]{Annie Grace Wang}
\address{Department of Mathematics, University of Pittsburgh, PA 15260.}
\email{agw23@pitt.edu}

\begin{document}
\begin{abstract}
    We investigate conjugate flows in a two-layer fluid system with a flat bottom and a free surface. These are up- and downstream states consistent with simultaneous conservation of mass, momentum, and energy; and therefore candidate end-states of one-dimensional dissipationless bores. By choosing variables in a specific way, we are able to \emph{explicitly} parameterize \emph{all} physical non-trivial conjugate flows in the system, using only rational functions. This two-parameter family can be categorized according to whether the states are compatible with %
    \begin{enumerate*}[1)]
        \item a rising surface, and falling internal interface; or
        \item vice versa
    \end{enumerate*}%
    . In particular, we rule out the existence of \emph{monotone} bore solutions, consistent with an earlier result due to \textcite{chen_16_stratified}. However, here we do so already at the level of conjugate flows. Separating these two distinct categories is a one-parameter family of critical trivial conjugate flows. The complete description obtained in the present paper serves as an important stepping stone towards finding actual realizing heteroclinic orbits connecting such conjugate states.
\end{abstract}
\maketitle

\section{Introduction}

Internal bores and hydraulic jumps in stratified fluids are robust manifestations of nonlinear wave dynamics, arising naturally in a range of geophysical and engineering contexts. They frequently appear when tidal currents interact with bottom topography, producing abrupt transitions in density stratification that propagate with near-permanent form. Such events have been observed in coastal oceans, estuaries, and laboratory experiments, with amplitudes that can span the entire depth of the water column~\cite{helfrich_06_long}. Because these waves modulate vertical mixing, nutrient transport, and even acoustic propagation~\cite{frank_05_experimental, whalen_20_internal}, understanding the conditions under which they form and persist remains an important problem.

A classical way to predict whether two uniform states can be connected by a dissipationless bore is through \emph{conjugate flow analysis}, first systematically developed for layered and continuously stratified fluids by \textcite{benjamin_66_internal, benjamin_71_unified}. In Benjamin's framework, the up- and downstream states are required to be consistent with conservation of mass, momentum (through ``flow force''), and energy. This is, of course, a necessary condition for the existence of such a bore. Since energy is conserved, conjugate flows simultaneously govern the broadening, or table-top, limit of large-amplitude solitary waves~\cite{turner_88_broadening}. It is worth noting that ``conjugate flow'' sometimes also refers more broadly to the energy-dissipative case.

The conjugate flow problem, and beyond, has since been studied more closely in a variety of stratified settings. For a two-layer fluid in a \emph{closed} channel, the problem is by now well understood, up to and including a global existence theory for internal gravity waves~\cite{amick_86_global, chen_26_global, chen_25_extreme}. For the natural extension to three layers,~\cite{lamb_00_conjugate} explores conjugate flows, while~\cite{rusaas_02_solitary} computes solitary waves numerically. Similarly,~\cite{lamb_98_conjugate} considers continuously stratified fluids. There is also a large array of results for model equations, for instance, a recent program to classify solitary waves in a system that retains the exact conjugate flows of the Euler equations~\cite{doak_22_internal, barros_26_internal}.

In the present work, we study conjugate flows for a two-layer fluid under a free surface, departing from the rigid-lid assumption above. While a rigid-lid may be satisfactory in many settings, it fails to capture the coupling between internal and surface waves occurring in, for instance, tidal bores. The question of conjugate flows has previously been examined most closely by \textcite[Appendix A]{dias_01_interfacial}, in connection with a model equation designed to capture generalized bores. Their description of the set of conjugate flows, however, is largely obtained numerically. To the best of our knowledge, therefore, even basic structural questions about the set of nontrivial conjugate flows remain open.

A rigorous existence theory for bores in this configuration accordingly also remains open, in contrast with the rigid-lid case. Moreover, as in~\cite{chen_16_stratified, sinambela_21_large}, the non-existence of certain bores enters compactness arguments in the constructions of solitary waves. This is one reason a complete description of the set of conjugate flows is of interest, beyond the actual bores themselves. Large-amplitude internal solitary waves under a free surface have been computed numerically~\cite{woolfrenden_11_solitary}, where broadening to a table-top is observed. This suggests the existence of corresponding bore solutions.

\subsection{Informal statement and outline}

Our main result is a complete, and entirely explicit, resolution of the non-trivial conjugate flows, in terms of dimensionless variables. Here, $H$ and $B$ are heights of the bottom layer, $D$ is the depth ratio, $R$ is the density ratio, $F$ is the (upstream) Froude number, and $x$, $y$ are velocity ratios. Their precise definitions are given in \cref{sec formulation}.

\begin{mainresult*}[Informal]
    \hypertarget{thm main-informal}{}

    The set $\Omega$ of nontrivial physical conjugate flows is the graph of an explicit map
    \[
        (x,y) \longmapsto \parn{B,R,D,H,F^2},
    \]
    where the components are rational functions, given in~\eqref{parameterization_hat}, on the parameter set
    \[
        \hat{\Omega} = \set*{(x,y) \in \R_+^2 : (x+y-2)(y-1) < 0}.
    \]
    Moreover:
    \begin{enumerate}[(i)]
        \item \label{thm main components} $\hat{\Omega}$ has exactly two connected components,
              \[
                  \hat{\Omega}_{D \lessgtr 1} \ceq \set*{(x,y) \in \R_+^2 : 2-x \lessgtr y \lessgtr 1 \lessgtr x},
              \]
              on which $D \lessgtr 1$ and $B \gtrless H$, respectively. The components are separated by the single point $(1,1)$, which, after rational reparameterization, can be seen to correspond to a one-parameter family $\Gamma_{D=1}$ of critical trivial conjugate flows.
        \item \label{thm main monotone} There is a further explicit rational parameterization in terms of the pair $(t,R)$, where $t = (1 - y)/(x - 1)$, that is strictly monotone with respect to $t$. Specifically, $D$, $H$, $B$, $F^2$, $y^2 F^2$, and $y$ are strictly decreasing with respect to $t$, while $x$ and $x^2 F^2$ are strictly increasing. As a particular consequence, $\Omega_* \ceq \Omega \cup \Gamma_{D=1}$ is in one-to-one correspondence with $(H,R) \in (0,1)^2$, proving a claim of~\cite{dias_01_interfacial}. The latter correspondence is not explicit.
        \item The monotonicity furnished by~\ref{thm main monotone} immediately yields sharp, $R$-independent bounds on all flow quantities. For instance, $D < \tfrac{3}{2}$, $\tfrac{4}{9} < H, B < 1$, and $F^2 < 1$ on $\hat{\Omega}_{D>1}$; while $H,B,F^2 < \tfrac{1}{2}$ on $\hat{\Omega}_{D<1}$.
    \end{enumerate}
\end{mainresult*}

Several features of this result deserve emphasis. First, it is a \emph{complete} description; second, the physical parameter region is a \emph{polygonal} domain, cut out by a strikingly simple sign condition; third, the parameterization is by \emph{rational} functions; and fourth, the \emph{apparent} singular point is ameliorated after reparameterization. Together, they make $\Omega$ extraordinarily amenable to further analysis. We expect this to be decisive in the follow-up problem of constructing a global set of actual heteroclinic orbits connecting the pairs of conjugate states.

As a consequence of \hyperlink{thm main-informal}{Main Result}~\ref{thm main components}, every nontrivial conjugate flow satisfies $(D - 1)(B - H) < 0$. That is, the displacement of the free surface and internal interface, or pycnocline, are always opposite of each other. This is consistent both with broadening exclusively being observed for ``out-of-phase'' solitary waves in~\cite{woolfrenden_11_solitary}, and with the non-existence of monotone bores~\cite{chen_16_stratified, sinambela_21_large}.

Our approach to \hyperlink{thm main-informal}{Main Result} relies on a choice of independent variables that is not the most obvious one, and we believe the mechanism is of independent interest. In principle, it is possible to eliminate all but two of the original variables using successive resultants. However, the degree and complexity of the expressions involved grow rapidly, which is precisely why the existing literature stops at a numerical description.

We instead exploit that the system of five polynomial equations in seven variables carries two complementary pieces of linear structure. First, when $x$ and $y$ are used as unknowns, three of the equations are \emph{linear} in $(D,H,F^2)$, see~\eqref{conservation_of_mass_nondim_y}--\eqref{bernoulli_free_nondim}, so that $(D,H,F^2)$ can be solved for in terms of rational functions in $(B,x,y)$. The lone degenerate case is $x = 1$, but yields only trivial solutions by \cref{trivial_lemma}. Second, the two remaining equations are linear in the density ratio $R$, so a \emph{single} simple resultant with respect to $R$ closes the elimination. That resultant is the heart of the matter, and here we were rather fortunate: it splits completely with respect to $B$, see~\eqref{resultant}. Moreover, all but one factor either force a trivial solution by \cref{trivial_lemma}, or has a sign for physically relevant values. The sole surviving factor enables us to solve for $B$ in terms of $(x,y)$, and by back-substitution, the remaining four expressions. We have no satisfactory explanation for precisely why the factorization is so clean in these variables.

The singularity of the parameterization at $(x,y) = (1,1)$ is subsequently resolved by passing to variables $(t,X)$ adapted to the pencil of lines through that point. Explicitly, $x \eqc 1 + X$, and $y \eqc 1 - tX$, with $t$ therefore being a slope. In these variables, the critical trivial conjugate flows appear as a line segment separating the two components, and the parameterization extends smoothly up to the boundary. Trading $X$ for $R$, which \emph{still} preserves rationality, and is natural in applications, gives the monotonicity in $t$ recorded in \hyperlink{thm main-informal}{Main Result}~\ref{thm main monotone}.

The paper is organized as follows. In \cref{sec formulation}, we set up the conjugate flow system for a two-layer fluid with a free surface, non-dimensionalize it, and isolate the trivial solutions. \cref{sec sol} carries out the elimination procedure described above, and determines the physical parameter region. Next, in \cref{sec reparam}, we develop the two reparameterizations, prove the claimed monotonicity, and derive the sharp bounds. Finally, \cref{sec uniqueness} establishes the uniqueness of a particular special family of conjugate flows.

\section{Formulation of the problem}
\label{sec formulation}

\subsection{Setting and notation}

We consider two immiscible incompressible fluid layers of different densities $\rho_1 > \rho_2 > 0$, under the influence of gravity, whose acceleration is denoted by $g$. The flow is assumed to be uniform far up- and downstream in each layer, and takes place in an open channel with a flat bottom. Both the pycnocline and surface are therefore free. In the frame moving with the bore, we take the upstream speed to be $c$ in both layers, while the downstream speeds are denoted $\mo{c}_j$. The up- and downstream layer heights are $h_j$ and $\mo{h}_j$, respectively. In general, we use a ring to denote downstream quantities, and the layers are numbered from bottom to top. See \cref{fluid_domain}.

\begin{figure}[htb]
    \centering
    \tikzsetnextfilename{fluid_domain}
    \includegraphics[width=0.7\textwidth]{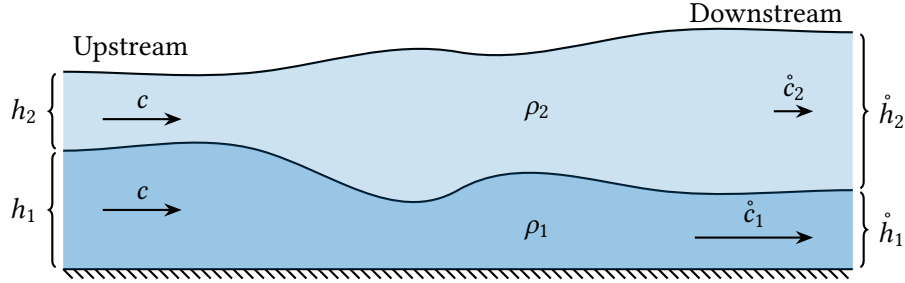}
    \caption{A bore propagating in a two-layer fluid with a flat bottom and a free surface.}
    \label{fluid_domain}
\end{figure}

\subsection{Equations}

Conservation of mass within each layer corresponds to the two equations
\begin{equation}
    \label{conservation_of_mass}
    \rho_j c h_j = \rho_j \mo{c}_j \mo{h}_j, \qquad j = 1,2,
\end{equation}
while the Bernoulli equation reads
\begin{equation}
    \label{bernoulli_equation}
    \tfrac{1}{2}\rho_j v^2 + \rho_j g z + p_j = Q_j
\end{equation}
in each layer. Here $Q_j$ is the Bernoulli constant for each layer, $v$ is the fluid velocity, and $z$ is the height above the bottom. If we normalize with zero pressure on the free surface, then
\begin{equation}
    \label{bernoulli_free}
    Q_2 = \tfrac{1}{2} \rho_2c^2 + \rho_2 g (h_1 + h_2) = \tfrac{1}{2} \rho_2 \mo{c}_2^2 + \rho_2 g (\mo{h}_1 + \mo{h}_2),
\end{equation}
by evaluating~\eqref{bernoulli_equation} up- and downstream. This results in the expected hydrostatic pressure
\begin{equation}
    \label{pressure_upper}
    p_2 = \case*{\begin{matrix*}[l]
        \rho_2 g(h_1 + h_2 - z)           & \text{upstream},   \\
        \rho_2 g(\mo{h}_1 + \mo{h}_2 - z) & \text{downstream},
    \end{matrix*}} 
\end{equation}
in the upper layer.

Next, continuity of the pressure across the internal interface yields
\begin{equation}
    \label{bernoulli_internal}
    Q_1 = \tfrac{1}{2} \rho_1 c^2 + \rho_1 g h_1 + \rho_2 g h_2 = \tfrac{1}{2}\rho_1 \mo{c}_1^2 + \rho_1 g \mo{h}_1 + \rho_2 g \mo{h}_2,
\end{equation}
where we have used~\eqref{pressure_upper}. The pressure in the lower layer is therefore given by
\begin{equation}
    \label{pressure_lower}
    p_1(z) = \case*{\begin{matrix*}[l]
                 \rho_1 g (h_1 - z) + \rho_2 g h_2           & \text{upstream},   \\
                 \rho_1 g (\mo{h}_1 - z) + \rho_2 g \mo{h}_2 & \text{downstream},
             \end{matrix*}}
\end{equation}
which again is hydrostatic.

From~\eqref{pressure_upper} and~\eqref{pressure_lower} we see that the total pressure force is
\[
    \brak[\big]{\rho_2 g \parn[\big]{\tfrac{1}{2}h_2}}h_2 + \brak[\big]{\rho_1 g \parn[\big]{\tfrac{1}{2}h_1} + \rho_2 g h_2}h_1 = \tfrac{1}{2}\rho_1 g h_1^2 + \tfrac{1}{2} \rho_2 g h_2^2 + \rho_2 g h_1 h_2
\]
upstream, and
\[
    \tfrac{1}{2}\rho_1 g \mo{h}_1^2 + \tfrac{1}{2} \rho_2 g \mo{h}_2^2 + \rho_2 g \mo{h}_1 \mo{h}_2
\]
downstream. Conservation of momentum therefore leads to the equation
\begin{multline}
    \label{conservation_of_momentum}
    \parn{\rho_1 h_1 c^2 + \rho_2 h_2 c^2} + \parn[\big]{\tfrac{1}{2}\rho_1 g h_1^2 + \tfrac{1}{2}\rho _2 g h_2^2 + \rho_2 g h_1 h_2}                                                 \\
    = \parn{\rho_1 \mo{h}_1 \mo{c}_1^2 + \rho_2 \mo{h}_2 \mo{c}_2^2} + \parn[\big]{\tfrac{1}{2} \rho_1 g \mo{h}_1^2 + \tfrac{1}{2} \rho_2 g \mo{h}_2^2 + \rho_2 g \mo{h}_1 \mo{h}_2}.
\end{multline}
Conjugate flows in this setting are simultaneous solutions to the system of equations consisting of conservation of mass in each layer~\eqref{conservation_of_mass}; continuity of pressure across each interface~\eqref{bernoulli_free},~\eqref{bernoulli_internal}; and conservation of momentum~\eqref{conservation_of_momentum}.

\subsection{Non-dimensionalization}

We introduce the dimensionless variables
\begin{equation}
    \label{upcvar}
    \begin{gathered}
        H_j = \frac{h_j}{h_1 + h_2},\quad \mo{H}_j = \frac{\mo{h}_j}{h_1 + h_2}, \quad \mo{F}_j^2 = \frac{\mo{c}_j^2}{g(h_1 + h_2)}, \\
        R = \frac{\rho_2}{\rho_1},\quad F^2 = \frac{c^2}{g(h_1 + h_2)},\quad D = \frac{\mo{h}_1 + \mo{h}_2}{h_1 + h_2},
    \end{gathered}
\end{equation}
where $0 < H_j < 1$, and $0 < \mo{H}_j < D$. Here $H_j, \mo{H}_j$ are dimensionless layer heights, $0 < R < 1$ is the density ratio, $F > 0$ is the upstream Froude number, $\mo{F}_j > 0$ are downstream Froude numbers, and $D > 0$ is the dimensionless downstream depth. Clearly, we have the relations
\begin{equation}
    \label{layer_height_sums}
    H_1 + H_2 = 1, \qquad \mo{H}_1 + \mo{H}_2 = D
\end{equation}
relating the various layer heights.

Substituting~\eqref{upcvar} into~\eqref{conservation_of_mass},~\eqref{bernoulli_free},~\eqref{bernoulli_internal}, and~\eqref{conservation_of_momentum} yields the non-dimensionalized system
\begin{align*}
    F H_j                                                       & = \mo{F}_j \mo{H}_j                                                                                        \\
    \tfrac{1}{2} F^2 + 1                                        & = \tfrac{1}{2} \mo{F}_2^2 + D                                                                              \\
    \tfrac{1}{2} F^2 + H_1 + R H_2                              & = \tfrac{1}{2}\mo{F}_1^2 + \mo{H}_1 + R \mo{H}_2                                                           \\
    F^2(H_1 + R H_2) + \tfrac{1}{2}\parn*{H_1^2 + R H_2(1+H_1)} & = \mo{F}_1^2 \mo{H}_1 + R \mo{F}_2^2 \mo{H}_2 + \tfrac{1}{2}\parn*{\mo{H}_1^2 + R \mo{H}_2(D + \mo{H}_1)},
\end{align*}
which, if we introduce
\[
    H \ceq H_1, \quad B \ceq \mo{H}_1, \quad x \ceq \frac{\mo{F}_2}{F}, \quad y \ceq \frac{\mo{F}_1}{F}
\]
and eliminate $H_2$, $\mo{H}_2$ using~\eqref{layer_height_sums}, reads
\begin{subequations}
    \label{nondim}
    \begin{align}
        H                                                        & = y B \label{conservation_of_mass_nondim_y}                                                                     \\
        1 - H                                                    & = x (D - B) \label{conservation_of_mass_nondim_x}                                                               \\
        \tfrac{1}{2} F^2 + 1                                     & = \tfrac{1}{2} x^2 F^2 + D \label{bernoulli_free_nondim}                                                        \\
        \tfrac{1}{2} F^2 + H + R (1 - H)                         & = \tfrac{1}{2}y^2 F^2 + B + R (D - B) \label{bernoulli_internal_nondim}                                         \\
        F^2(H + R (1 - H)) + \tfrac{1}{2}\parn*{H^2 + R (1-H^2)} & = F^2\parn*{y^2 B + x^2R (D-B)} + \tfrac{1}{2}\parn*{B^2 + R (D^2-B^2)} \label{conservation_of_momentum_nondim}
    \end{align}
\end{subequations}
in terms of $(D,H,B,R,F^2,x,y)$. This is the system we will investigate.

Physical non-stationary solutions to~\eqref{nondim} are those that satisfy
\begin{alignat}{-1}
    0 & < H < 1, \quad 0 < B < D, \quad & 0 & < R < 1,\quad & x,y   & > 0, \quad & F^2 & > 0,\label{physical}
    \intertext{with any choice}
    0 & < H = B < 1 = D,                & 0 & < R < 1,      & x = y & = 1,       & F^2 & > 0 \label{trivial}
\end{alignat}
being a \emph{trivial} solution. If one equality holds here, then they all hold:

\begin{lemma}
    \label{trivial_lemma}
    If either $x = y$, $x = 1$, $y = 1$, $D = 1$, or $H = B$ for a physical solution, then it is necessarily trivial.
\end{lemma}
\begin{proof}
    Observe that we may write~\eqref{bernoulli_free_nondim} and~\eqref{bernoulli_internal_nondim} as
    \begin{align}
        \tfrac{1}{2}(1 - x^2) F^2 & = D - 1, \label{bernoulli_free_nondim_3}                         \\
        \tfrac{1}{2}(1 - y^2) F^2 & = (1 - R)(B - H) + R(D - 1), \label{bernoulli_internal_nondim_3}
    \end{align}
    and that we can subtract them to obtain
    \begin{equation}
        \label{bernoulli_subtracted}
        \tfrac{1}{2}(x^2 - y^2)F^2 + (1 - R)(D - 1) = (1 - R)(B - H).
    \end{equation}
    It is straightforward to show that any of the equalities imply the other ones:
    \begin{description}
        \item[$x = y \implies x = 1$] If $x = y$, we get $D - 1 = B - H$ from~\eqref{bernoulli_subtracted}, or equivalently $D - B = 1 - H$. We use this to obtain $x = 1$ from~\eqref{conservation_of_mass_nondim_x}.
        \item[$x = 1 \iff D = 1$] This is immediate from~\eqref{bernoulli_free_nondim_3}.
        \item[$D = 1 \implies H = B$] If $D = 1$, then $x = 1$. We subsequently find $H = B$ from~\eqref{conservation_of_mass_nondim_x}.
        \item[$H = B \iff y = 1$] Immediate from~\eqref{conservation_of_mass_nondim_y}.
        \item[$y = 1 \implies x = y$] If $y = 1$, then $H = B$. We conclude from~\eqref{bernoulli_internal_nondim_3} that $D = 1$, and therefore that $x = 1$. \qedhere
    \end{description}
\end{proof}

\section{Solving for nontrivial conjugate flows in terms of \texorpdfstring{$(x,y)$}{(x,y)}}
\label{sec sol}

We define the set
\begin{equation}
    \label{nontrivial_manifold}
    \Omega \ceq \set*{(D,H,B,R,F^2,x,y) \in \R^7 : \text{\eqref{nondim} and~\eqref{physical} hold, but~\eqref{trivial} does \emph{not}}}
\end{equation}
of physical nontrivial conjugate flows. By \cref{trivial_lemma} we have
\begin{equation}
    \label{nontrivial_inequalities}
    x \neq y, \quad x,y,D \neq 1, \quad H \neq B
\end{equation}
for every such conjugate flow, and we may therefore work under this assumption.

\subsection{Eliminating \texorpdfstring{$D$}{D}, \texorpdfstring{$H$}{H} and \texorpdfstring{$F^2$}{F²}}

The subsystem~\eqref{conservation_of_mass_nondim_y} --~\eqref{bernoulli_free_nondim} is linear in $(D,H,F^2)$, and equivalent to
\begin{subequations}
    \label{HDF_Bxy}
    \begin{align}
        H                         & = y B \label{H_Bxy}                 \\
        x D                       & = 1 + (x - y)B \label{D_Bxy}        \\
        \tfrac{1}{2}x(x^2 - 1)F^2 & = (x - 1) - (x - y)B, \label{F_Bxy}
    \end{align}
\end{subequations}
since $x \neq 0$ for physical solutions. Consequently, we may express $(D,H,F^2)$ in terms of $(B,x,y)$, since also $x \neq 1$. Once inserted into~\eqref{bernoulli_internal_nondim} and~\eqref{conservation_of_momentum_nondim}, we are left with two polynomial equations in the four remaining variables $(B,R,x,y)$.

\subsection{Reduction to a resultant problem}

The aforementioned system in $(B,R,x,y)$ is linear in $R$, and reads
\begin{equation}
    \label{system_BRxy}
    \begin{aligned}
        a_1 R + a_0 & = 0, \qquad \begin{aligned}
                                      a_1 & \ceq (x - 1)^2(x + 1)(1 - By)                             \\
                                      a_0 & \ceq (y - 1)\brak*{B \parn*{x^3+(x-1)y-y^2} - (x-1)(y+1)}
                                  \end{aligned} \\
        b_1 R + b_0 & = 0, \qquad \begin{aligned}
                                      b_1 & \ceq (x - 1)^2(1 - By)\brak*{B\parn*{2x+m(x)y}+(x-1)^2}                 \\
                                      b_0 & \ceq (y - 1)x B\brak*{ B \parn*{x(x^2-1)+x(3+x^2)y - 4y^2} - 4 y(x-1)},
                                  \end{aligned}
    \end{aligned}
\end{equation}
where we have defined
\[
    m(x) \ceq x^2 + 2 x - 1.
\]
Its resultant with respect to $R$ factors surprisingly nicely with respect to $B$ as
\begin{equation}
    \label{resultant}
    a_1 b_0 - b_1 a_0 = (x - 1)^2(y - 1)\parn[\big]{B(x-y)-(x-1)}(1 - By)\brak*{B(x-y)\hat{s}(x,y)-(x-1)^2(y+1)},
\end{equation}
with
\[
    \hat{s}(x,y) \ceq m(x) y + n(x), \qquad n(x) \ceq x^3 - x^2 - x - 1.
\]

\subsection{Analysis of the resultant}

The system~\eqref{system_BRxy} has a solution if and only if the resultant in~\eqref{resultant} vanishes. We therefore investigate each factor for roots leading to nontrivial solutions.

\begin{description}
    \item[$(x - 1)$ and $(y - 1)$] Excluded by~\eqref{nontrivial_inequalities}
    \item[$B(x - y) - (x - 1)$] Since~\eqref{D_Bxy} implies that $xD = 1 + (x - 1) = x$ in this case, we get $D = 1$, which again is excluded by~\eqref{nontrivial_inequalities}
    \item[$1 - By$] For physical solutions $1 - By = 1 - H > 0$ by~\eqref{H_Bxy}.
\end{description}

The only remaining possibility for nontrivial solutions is thus that
\[
    B(x - y)\hat{s}(x,y) = (x - 1)^2(y + 1)
\]
is satisfied. For nontrivial solutions $\hat{s}$ must therefore necessarily be nonzero, and
\begin{subequations}
    \label{parameterization_hat}
    \begin{equation}
        \label{B_hat}
        B = \hat{B}(x,y) \ceq \frac{(x - 1)^2(y + 1)}{(x - y)\hat{s}(x,y)},
    \end{equation}
    which inserted back into the first equation in~\eqref{system_BRxy} yields
    \[
        R(x - 1)\hat{w}(x,y) = (y - 1)^3(y + 1),
    \]
    where
    \[
        \hat{w}(x,y) \ceq 2x y^2 - 2(x + 1)y - n(x).
    \]

    Like before, $\hat{w}$ must necessarily be nonzero at a nontrivial solution, so we conclude that
    \begin{equation}
        \label{R_hat}
        R = \hat{R}(x,y) \ceq \frac{(y - 1)^3 (y + 1)}{(x - 1)\hat{w}(x,y)},
    \end{equation}
    and using~\eqref{HDF_Bxy} with~\eqref{B_hat} we finally find
    \begin{equation}
        D = \hat{D}(x,y) \ceq \frac{2 x y + x^2 - 3}{\hat{s}(x,y)}, \quad H = \hat{H}(x,y) \ceq \frac{(x - 1)^2(y + 1)y}{(x - y)\hat{s}(x,y)}, \quad F^2 = \hat{\Phi}(x,y) \ceq \frac{2(x + y - 2)}{\hat{s}(x,y)}. \label{DHPhi_hat}
    \end{equation}
\end{subequations}
That is, the set $\Omega$ in~\eqref{nontrivial_manifold} is a submanifold of the graph of the rational functions in~\eqref{parameterization_hat}. It remains to determine the exact parameter region. This will turn out to have a pleasantly simple description.

\subsection{The physical region in terms of \texorpdfstring{$(x,y)$}{(x,y)}}

Let us define
\begin{equation}
    \label{physical_region_set}
    \hat{\Omega} \ceq \set*{(x,y) \in \R_+^2 : \text{\eqref{parameterization_hat} are well-defined, and satisfy the inequalities in~\eqref{physical}}},
\end{equation}
which by the above discussion is exactly the parameter region corresponding to $\Omega$. To facilitate working with the upper bounds in~\eqref{physical}, we record that
\begin{equation}
    \label{minus_hat}
    \hat{D} - \hat{B} = - \frac{\hat{w}(x,y)}{(x - y)\hat{s}(x,y)}, \quad 1 - \hat{H} = - \frac{x\hat{w}(x,y)}{(x - y)\hat{s}(x,y)}, \quad 1 - \hat{R} = - \frac{(x - y)^2 (x + y - 2)(x + y)}{(x - 1)\hat{w}(x,y)}.
\end{equation}
The constraints can be translated into conditions on the factors appearing in~\eqref{parameterization_hat} and~\eqref{minus_hat}. Explicitly
\begin{align*}
    (x - y)\hat{s}(x,y)              & > 0, \tag{$\hat{B} > 0$ and $\hat{H} > 0$}       \\
    \hat{w}(x,y) (x - y)\hat{s}(x,y) & < 0, \tag{$\hat{B} < \hat{D}$ and $\hat{H} < 1$} \\
    \hat{w}(x,y)(x - 1)(y - 1)       & > 0, \tag{$\hat{R} > 0$}                         \\
    \hat{w}(x,y)(x + y - 2)(x - 1)   & < 0, \tag{$\hat{R} < 1$}                         \\
    \hat{s}(x,y)(x + y - 2)          & > 0, \tag{$\hat{\Phi} > 0$}
\end{align*}
which may be summarized as
\[
    (x + y - 2)(y - 1) < 0, \quad \hat{w}(x,y) < 0, \quad \text{and} \quad \sgn(\hat{s}(x,y)) = \sgn(x - 1) = \sgn(x - y) = \sgn(x + y - 2).
\]

In fact, since one has the identities
\begin{alignat*}{-1}
    \hat{w}(x,y) & = (x + y - 2)\parn*{(x+y+2)(y-1) - (x^2-1)} - (y + 1)(y - 1)^2, \\
    \hat{s}(x,y) & = (1 + x + x^2 + y)(x + y - 2) - (y^2 - 1),                     \\
    x - y        & = (x + y - 2) - 2(y - 1),                                       \\
    x - 1        & = (x + y - 2) - (y - 1),
\end{alignat*}
the physical region in~\eqref{physical_region_set} can equivalently be described as simply
\[
    \hat{\Omega} = \set*{(x,y) \in \R_+^2 : (x+y-2)(y-1) < 0},
\]
which is natural to split as
\[
    \hat{\Omega} = \hat{\Omega}_{D>1} \cup \hat{\Omega}_{D<1},
\]
with
\[
    \hat{\Omega}_{D \lessgtr 1} \ceq \set*{(x,y) \in \R_+^2 : 2-x \lessgtr y \lessgtr 1 \lessgtr x}.
\]
This also induces components $\Omega_{D \lessgtr 1}$ of $\Omega$. The naming is motivated by the following dichotomy, which follows immediately from~\eqref{bernoulli_free_nondim_3} and~\eqref{bernoulli_internal_nondim_3}.

\begin{lemma}
    $D \lessgtr 1$ and $B \gtrless H$ on $\Omega_{D \lessgtr 1}$.
\end{lemma}
\begin{remark}
    In particular, one has $(D - 1)(B - H) < 0$ for all nontrivial conjugate flows. The surface and pycnocline therefore always displace in diametrically opposite directions. This is consistent with our expectation from the established \emph{non}existence of stratified monotone bores~\cite{chen_16_stratified, sinambela_21_large}. The present analysis shows that, at least for the special case of two layers, the existence is precluded already at the conjugate flow level. We also note that the property is in line with the observation that broadening only seems to occur for ``out-of-phase'' solitary waves in~\cite{woolfrenden_11_solitary}.
\end{remark}

The boundary of $\hat{\Omega}$ can be decomposed as
\[
    \partial \hat{\Omega} = \hat{\Gamma}_{H=0} \cup \hat{\Gamma}_{H=1} \cup \hat{\Omega}_{D=1} \cup \hat{\Gamma}_{R=0} \cup \hat{\Gamma}_{R=1},
\]
where
\[
    \begin{aligned}
        \hat{\Gamma}_{H=0} & \ceq [2,\infty) \times \set{0}, \\
        \hat{\Gamma}_{H=1} & \ceq \set{0} \times [1,2],
    \end{aligned}
    \qquad \hat{\Gamma}_{D=1} \ceq \set{(1,1)}, \qquad
    \begin{aligned}
        \hat{\Gamma}_{R=0} & \ceq \overline{\R}_+ \times \set{1},            \\
        \hat{\Gamma}_{R=1} & \ceq \set{(x,y) \in \overline{\R}_+^2 : x+y=2},
    \end{aligned}
\]
as illustrated in \cref{fig:parameter_region}. We also define
\[
    \hat{\Omega}_* \ceq \hat{\Omega} \cup \hat{\Gamma}_{D=1},
\]
which includes the point corresponding to critical trivial conjugate flows.

\begin{figure}[htb]
    \centering
    \tikzsetnextfilename{parameter_region}
    \includegraphics[width=.7\textwidth]{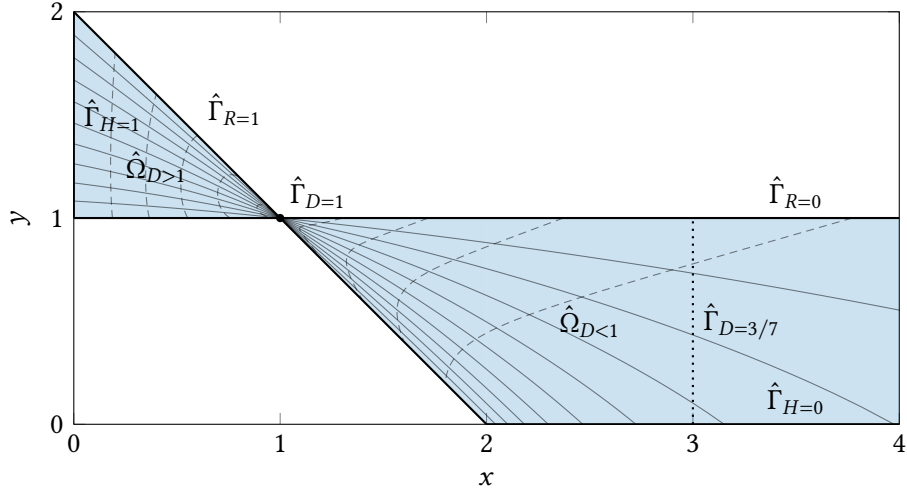}
    \caption{The physical parameter region. We restrict our view to $x < 4$ here, but $\hat{\Omega}_{D<1}$ is of course unbounded. The level curves are $R = (k/10)^3$ (solid) and $H = k/10$ (dashed) for $k = 1,\dots,9$, ordered bottom to top at the $y$-axis, and right to left, respectively.}
    \label{fig:parameter_region}
\end{figure}

\subsection{A special family}
\label{special_section}

If we compute the derivatives
\begin{equation}
    \label{D_hat_Phi_hat_derivatives}
    \partial_y \hat{\Phi} = -\frac{2(x - 3)(x - 1)}{\hat{s}(x,y)^2}, \qquad \partial_y \hat{D} = \frac{(x + 1)(x - 1)^2(x - 3)}{\hat{s}(x,y)^2},
\end{equation}
we notice that there is change in monotonicity of $\hat{\Phi}$ and $\hat{D}$ with respect to $y$ at $x = 3$. In particular, $\hat{\Phi}$ and $\hat{D}$ are both constant along the line $x = 3$. Specifically,
\[
    \hat{D} = \frac{3}{7}, \quad \hat{H} = \frac{2y}{7(3 - y)}, \quad \hat{B} = \frac{2}{7(3 - y)}, \quad \hat{R} = \frac{(1 - y)^3}{4(7 - 3y)}, \quad \hat{\Phi} = \frac{1}{7},
\]
at each point $(3,y)$ on
\[
    \hat{\Gamma}_{D=3/7} \ceq \set{3} \times (0,1).
\]
We prove in \cref{special_uniqueness} that this is the only family on $\Omega$ for which two or more variables are held constant.

\section{Reparameterization}
\label{sec reparam}

\subsection{The pencil of lines through \texorpdfstring{$(1,1)$}{(1,1)}}

If we define variables $(t,X)$ through
\begin{equation}
    \label{pencil_variables}
    x = \tilde{x}(t,X) \ceq 1 + X, \qquad y = \tilde{y}(t,X) \ceq 1 - t X,
\end{equation}
then $\hat{\Omega}_{D>1}$ and $\hat{\Omega}_{D<1}$ map to
\[
    \tilde{\Omega}_{D>1} = (0,1) \times (-1,0), \qquad \text{and} \qquad \tilde{\Omega}_{D<1} = \set*{(t,X) \in (0,1) \times \R_+ : t X < 1},
\]
respectively. From~\eqref{parameterization_hat}, we obtain
\begin{equation}
    \label{solution_tilde}
    \begin{aligned}
        D & = \tilde{D}(t,X) \ceq \frac{2(2 - t) + (1 - 2t)X}{\tilde{s}(t,X)},     \\
        H & = \tilde{H}(t,X) \ceq \frac{(2 - t X)(1 - tX)}{(1 + t)\tilde{s}(t,X)}, \\
        B & = \tilde{B}(t,X) \ceq \frac{2 - tX}{(1 + t)\tilde{s}(t,X)},
    \end{aligned}
    \qquad
    \begin{aligned}
        R   & = \tilde{R}(t,X) \ceq \frac{(2 - tX)t^3}{\tilde{w}(t,X)}, \\
        F^2 & = \tilde{\Phi}(t,X) \ceq \frac{2(1 - t)}{\tilde{s}(t,X)},
    \end{aligned}
\end{equation}
where
\[
    \tilde{s}(t,X) \ceq 2(2 - t) + (3 - 4t)X + (1 - t)X^2, \qquad \tilde{w}(t,X) \ceq 2(1 + t(1 - t)) - (2t^2 - 1)X.
\]

\begin{figure}[htb]
    \centering
    \tikzsetnextfilename{parameter_region_tilde}
    \includegraphics[width=0.7\textwidth]{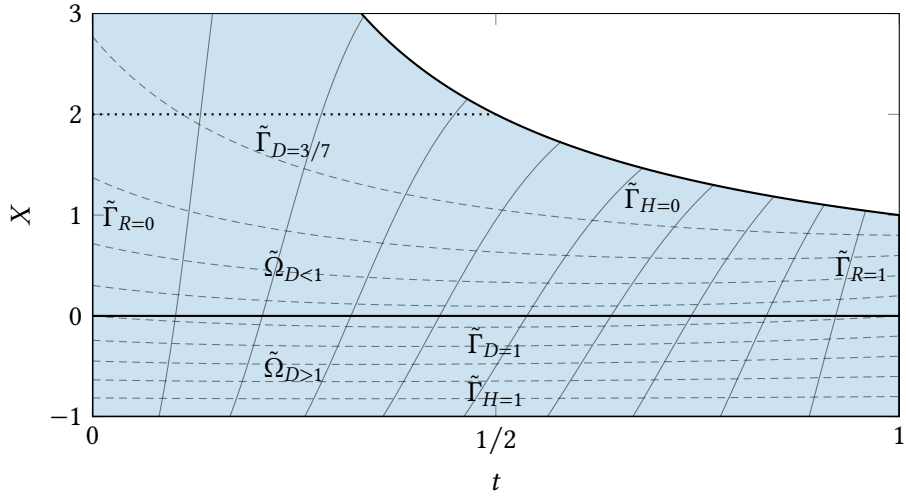}
    \caption{The physical parameter region in $(t,X)$, with the same level curves as in \cref{fig:parameter_region}.}
    \label{fig:parameter_region_tilde}
\end{figure}

Since $\tilde{s}$ is linear in $t$, and
\begin{alignat*}{-1}
    \tilde{s}(0,X)   & = (X + 2)(X + 1) + 1 > 1, \qquad         & (X > -1),     \\
    \tilde{s}(1,X)   & = 2 - X > 1,                             & (-1 < X < 1), \\
    \tilde{s}(t,1)   & = 8 - 7t > 1,                            & (0 < t < 1),  \\
    \tilde{s}(1/X,X) & = X + 2\parn[\Big]{1-\frac{1}{X^2}} > 1, & (X > 1),
\end{alignat*}
we have $\tilde{s} > 1$ on the larger
\[
    \tilde{\Omega}_* \ceq \set*{(t,X) \in (0,1) \times (-1,\infty) : t X < 1}.
\]
Similarly, $\tilde{w}$ is linear in $X$, and
\[
    \tilde{w}(t,-1) = 1 + 2t > 1, \qquad \tilde{w}(t,1/t) = \frac{1}{t} + 2(1 - t^2) > 1
\]
for all $t \in (0,1)$, so $\tilde{w} > 1$ on the same set.

The parameterized family in~\eqref{solution_tilde} therefore extends smoothly to $\tilde{\Omega}_*$, and even its boundary. The added line segment
\[
    \tilde{\Gamma}_{D=1} \ceq (1,0) \times \set{0},
\]
which was collapsed to a point in the previous parameterization, corresponds to a one-parameter family of critical trivial solutions $\Gamma_{D=1}$ separating $\Omega_{D>1}$ and ${\Omega}_{D<1}$. Explicitly
\[
    \tilde{D} = 1, \quad \tilde{H} = \tilde{B} = \frac{1}{(2 - t)(1 + t)}, \quad \tilde{R} = \frac{t^3}{1 + t(1 - t)}, \quad \tilde{\Phi} = \frac{1 - t}{2 - t}, \quad \tilde{x} = \tilde{y} = 1
\]
at $(t,0) \in \tilde{\Gamma}_{D=1}$, for all $t \in (0,1)$. Following our established notation, we write $\Omega_* \ceq \Omega \cup \Gamma_{D=1}$.

The boundary of $\tilde{\Omega}_*$ is
\[
    \partial \tilde{\Omega}_* = \tilde{\Gamma}_{H=0} \cup \tilde{\Gamma}_{H=1} \cup \tilde{\Gamma}_{R=0} \cup \tilde{\Gamma}_{R=1},
\]
see \cref{fig:parameter_region_tilde}, where
\begin{alignat*}{-1}
    \tilde{\Gamma}_{H=0} & \ceq \set{(t,X) \in (0,1] \times \R_+ : t X = 1}, & \qquad \tilde{\Gamma}_{H=1} & \ceq [0,1] \times \set{-1}, \\
    \tilde{\Gamma}_{R=0} & \ceq \set{0} \times [-1,\infty),                  & \tilde{\Gamma}_{R=1}        & \ceq \set{1} \times [-1,1],
\end{alignat*}
while the special family ends up on
\[
    \tilde{\Gamma}_{D=3/7} \ceq (0,1/2) \times \set{2}.
\]

\subsection{Replacing \texorpdfstring{$X$}{X} with \texorpdfstring{$R$}{R}}

In practice, it is more natural to fix $R$, so it makes sense to turn R into a parameter. This is still possible while retaining the use of rational functions. Moreover, crucially, we will see that this parameterization is, in fact, strictly monotone with respect to $t$.

Observe that it is straightforward to invert the equation for $R$ in~\eqref{solution_tilde} to obtain
\begin{equation}
    \label{X_breve}
    X = \breve{X}(t,R) \ceq \frac{2t^3 - 2R(1 + t(1 - t))}{t^4 - R(2 t^2 - 1)},
\end{equation}
with physical parameter set
\begin{align*}
    \breve{\Omega}_*       & \ceq \set*{(t,R) \in (0,1)^2 : R_{H=0}(t) < R < R_{H=1}(t)},
    \shortintertext{critical curve}
    \breve{\Gamma}_{D=1}   & \ceq \set*{(t,R) \in (0,1)^2 : R = R_{D=1}(t)},
    \shortintertext{and special family}
    \breve{\Gamma}_{D=3/7} & \ceq \set*{(t,R) \in \parn*{0,\tfrac{1}{2}} : R = R_{D=3/7}(t)},
\end{align*}
where we have introduced
\begin{equation}
    \label{R_functions}
    \begin{alignedat}{-1}
        R_{H=0}(t) & \ceq \frac{t^4}{1 + 2t(1 - t^2)}, \qquad & R_{D=3/7}(t) & \ceq \frac{t^3}{2 + 3t},       \\
        R_{H=1}(t) & \ceq \frac{t^3(2 + t)}{1 + 2t},          & R_{D=1}(t)   & \ceq \frac{t^3}{1 + (1 - t)t},
    \end{alignedat}
\end{equation}
as illustrated in \cref{fig:parameter_region_breve}. The boundary of $\breve{\Omega}_*$ is
\[
    \partial\breve{\Omega}_* = \breve{\Gamma}_{H=0} \cup \breve{\Gamma}_{H=1} \cup \breve{\Gamma}_{R=0} \cup \breve{\Gamma}_{R=1},
\]
with
\begin{alignat*}{-1}
    \breve{\Gamma}_{H=0} & \ceq \set{(t,R) \in [0,1]^2 : R = R_{H=0}(t)}, \qquad & \breve{\Gamma}_{R=0} & \ceq \set{(0,0)}, \\
    \breve{\Gamma}_{H=1} & \ceq \set{(t,R) \in [0,1]^2 : R = R_{H=1}(t)},        & \breve{\Gamma}_{R=1} & \ceq \set{(1,1)}.
\end{alignat*}

\begin{figure}[htb]
    \centering
    \tikzsetnextfilename{parameter_region_breve}
    \includegraphics[width=0.7\textwidth]{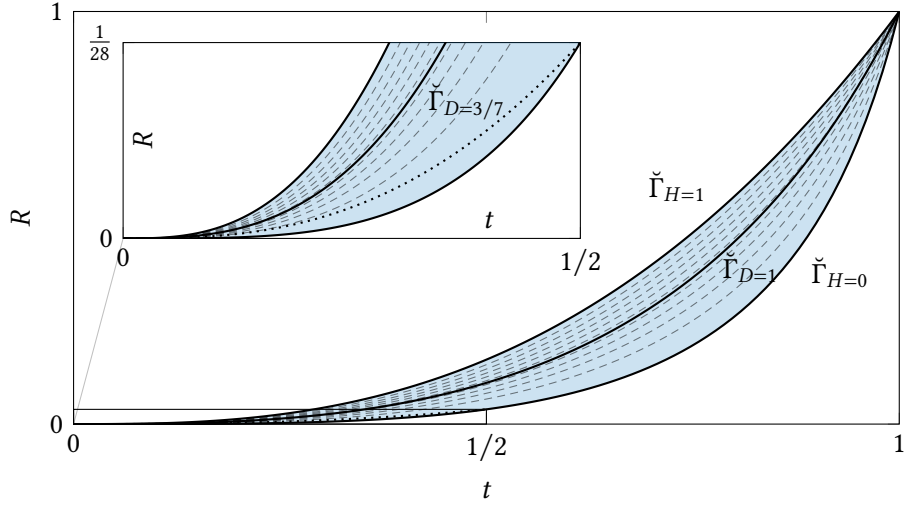}
    \caption{The physical parameter region in $(t,R)$. Only level curves for $H$ included, since the ones for $R$ are now straight lines.}
    \label{fig:parameter_region_breve}
\end{figure}

\begin{lemma}
    \label{breve_lemma}
    \leavevmode
    \begin{enumerate}[(i), beginpenalty=10000]
        \item The functions $R_{H=0}$, $R_{D=1}$, and $R_{H=1}$ in~\eqref{R_functions} satisfy $R(0) = 0$ and $R(1) = 1$, are strictly increasing on $[0,1]$, and satisfy
              \begin{alignat}{-1}
                  & R_{H=0}(t) < R_{D=1}(t) < R_{H=1}(t) < t^2 \qquad &  & \text{for all } t \in (0,1).\label{R_order}
                  \intertext{\item The function $R_{D=3/7}$ is strictly increasing on $[0,1/2]$, satisfies $R_{D=3/7}(0) = 0$ and $R_{D=3/7}(1) = 1/28$, and}
                  & R_{H=0}(t) < R_{D=3/7}(t) < R_{D=1}(t) \qquad     &  & \text{for all } t \in (0,1/2). \notag
              \end{alignat}
        \item Finally, the denominator in~\eqref{X_breve} is positive on $(0,1)^2$, and $\partial_t \breve{X}(t,R) > 0$ for all $(t,R) \in \breve{\Omega}_*$.\label{breve_lemma_part_3}
    \end{enumerate}
\end{lemma}
\begin{proof}
    Direct computation shows that
    \begin{gather*}
        R_{D=1}(t) - R_{H=0}(t) = \frac{(1 - t)t^3(1 + t)^2}{(1 + t(1 - t))(1 + 2t(1 - t^2))}, \qquad R_{H=1}(t) - R_{D=1} = \frac{(1 - t)t^3(1 + t)^2}{(1 + 2t)(1 + t(1 - t))}, \\
        t^2 - R_{H=1}(t) = \frac{(1 - t^2)t^2}{1 + 2t}
    \end{gather*}
    are all positive for all $t \in (0,1)$, demonstrating~\eqref{R_order}. Similarly,
    \[
        R_{H=0}'(t) = \frac{2(2 - t)t^3(1 + t)^2}{(1 + 2t(1 - t^2))},\quad R_{D=1}'(t) = \frac{(3 - t)t^2(1 + t)}{(1 + t(1 - t))^2},\quad R_{H=1}'(t) = \frac{6t^2(1 + t)^2}{(1 + 2t)^2}
    \]
    are all positive as well. The argument for $R_{D=3/7}$ on $(0,1/2)$ is similar, so we omit it.

    The denominator in~\eqref{X_breve} is positive on $(0,1)^2$, since it can be written as
    \[
        t^4 - R(2t^2 - 1) = (t^2 - R)^2 + R(1 - R),
    \]
    where both terms are positive. We next compute
    \[
        \partial_t \breve{X}(t,R) = \frac{-2R^2(1 + 2t(1 + t)) + 2Rt^2(3 + 4t + t^2 - 2t^3) - 2 t^6}{(t^4 - R(2t^2 - 1))^2},
    \]
    where we recognize that the numerator is strictly concave with respect to $R$. Since
    \[
        \partial_t \breve{X}(t,R_{H=0}(t)) = \frac{1}{(1 - t)t^2} > 0, \qquad \text{and} \qquad \partial_t \breve{X}(t,R_{H=1}(t)) = \frac{3}{(1 - t)t} > 0
    \]
    on $(0,1)$, we conclude that $\partial_t \breve{X}$ is positive on $\breve{\Omega}_*$.
\end{proof}

Substituting~\eqref{X_breve} into~\eqref{pencil_variables} and~\eqref{solution_tilde}, we obtain
\begin{equation}
    \label{solution_breve}
    \begin{alignedat}{-1}
        D & = \breve{D}(t,R) \ceq \frac{(R + t^3)(t^4 - R(2t^2 - 1))}{\breve{s}(t,R)}, \qquad & F^2 & = \breve{\Phi}(t,R) \ceq \frac{\parn{t^4-R(2t^2-1)}^2}{(1 + t)\breve{s}(t,R)}, \\
        H & = \breve{H}(t,R) \ceq \frac{R\parn[\big]{(1+2t(1-t^2))R-t^4}}{\breve{s}(t,R)},    & x   & = \breve{x}(t,R) \ceq \frac{t^3(2 + t) - R(1 + 2t)}{t^4 - R(2t^2 - 1)},        \\
        B & = \breve{B}(t,R) \ceq \frac{R\parn{t^4-R(2t^2-1)}}{\breve{s}(t,R)},               & y   & = \breve{y}(t,R) \ceq \frac{(1 + 2t(1 - t^2))R - t^4}{t^4 - R(2t^2 - 1)},
    \end{alignedat}
\end{equation}
where
\[
    \breve{s}(t,R) \ceq t^6(2 + t) - t^3(1 + 3t)R + (1 + 2t(1 - t^2))R^2
\]
is positive on $(0,1)^2$. This can, for instance, be seen from the identity
\[
    2 \breve{s}(t,R) = (1 + 3t)(R - t^3)^2 + R^2(1 - t)^2\parn*{1+3t+5t^2+3t^3+t^4(1-t)} + (1 - R^2)t^6(3 - t).
\]

\subsection{Monotonicity}

From \cref{breve_lemma} \ref{breve_lemma_part_3} we know that $\breve{x}$ and $\breve{y}$ are strictly increasing and decreasing, respectively, with respect to $t$. We will now see that \emph{every} variable is monotone with respect to $t$. It suffices to prove the following lemma.

\begin{lemma}
    \label{D_phi_derivative_sign}
    The signs $\partial_t \breve{\Phi} < 0$ and $\partial_t \parn{\breve{x}^2\breve{\Phi}} > 0$ hold on $\breve{\Omega}_*$.
\end{lemma}
\begin{proof}
    We first consider the sign of $\partial_t \breve{\Phi}$, which is the simpler of the two. Explicitly,
    \[
        \partial_t \breve{\Phi}(t,R) = - \frac{(t^2 - R)(t^4 - R(2t^2 - 1))\nu(t,R)}{(1 + t)^2 \breve{s}(t,R)^2},
    \]
    where
    \[
        \nu(t,R) \ceq - t^7(4 + 3t) + 2t^3(3 + 2t)(2 + 3t - 2t^3)R - \brak*{3 +4t(3+t^2(3-t^2))}R^2
    \]
    is strictly concave with respect to $R$ for $t \in (0,1)$. Since
    \[
        \nu(t,R_{H=0}(t)) = \frac{4(2 - t)(1 - t)t^7(1 + t)^5}{(1 - 2t(1 - t^2))^2} > 0, \qquad \nu(t,R_{H=1}(t)) = \frac{12(1 - t)t^6(1 + t)^5}{(1 + 2t)^2} > 0,
    \]
    for all $t \in (0,1)$, we conclude that $\nu > 0$ on $\breve{\Omega}_*$. As the two other factors in the numerator of $\partial_t \breve{\Phi}$ are also positive by \cref{breve_lemma}, we therefore have $\partial_t \breve{\Phi} < 0$ on $\breve{\Omega}_*$.

    Next, one may verify that
    \[
        \partial_t \parn*{\breve{x}^2\breve{\Phi}} = \frac{(1 + 2t)\parn*{R_{H=1}(t) - R}\mu(t,R)}{(1 + t)^2 \breve{s}(t,R)^2},
    \]
    with $\mu$ cubic in $R$. The sign of the first two factors in the numerator are certainly positive on $\breve{\Omega}_*$. The sign of $\partial_t (\breve{x}^2 \breve{\Phi})$ is therefore controlled by that of $\mu$. We do not reproduce the entire expression for $\mu$ here, merely that it is of the form
    \[
        \mu(t,R) = -(2 + t) t^9 + a_1(t) R + a_2(t) R^2 - \parn*{1+2t+6t^2+12t^3+8t^4}R^3.
    \]

    From the expression for $\mu$, we in particular see that $\mu(t,0) < 0$, and that $\mu(t,R) \to \mp \infty$ as $R \to \pm \infty$, for all $t \in (0,1)$. Moreover, it can be verified that
    \begin{align*}
        \mu(t,R_{H=0}(t)) & = \frac{4(1 - t)t^9 (1 + t)^6 (1 + 4t - t^2 - 2t^3)}{(1 + 2t(1 - t^2))^3} > 0, \\
        \mu(t,R_{H=1}(t)) & = \frac{24(1 - t)t^8 (1 + t)^6(2 + t)}{(1 + 2t)^3} > 0
    \end{align*}
    at the boundary. Indeed, the only factor that is not trivially positive satisfies
    \[
        1 + 4t - t^2 - 2t^3 > 1 + 4t - t - 2t = 1 + t
    \]
    for all $t \in (0,1)$. We conclude that $\mu$ has exactly three real roots with respect to $R$ for each $t \in (0,1)$, of which two lie to the left of $R_{H=0}(t)$, and one to the right of $R_{H=1}(t)$. Thus, $\mu(t,R)$ is positive on $\breve{\Omega}_*$.
\end{proof}

\begin{corollary}
    \label{monotonicity_corollary}
    For all $(t,R) \in \breve{\Omega}_*$, we have signs
    \[
        \partial_t \breve{D}, \partial_t \breve{H}, \partial_t \breve{B}, \partial_t \breve{\Phi}, \partial_t (\breve{y}^2 \breve{\Phi}), \partial_t \breve{y} < 0, \quad \text{and}\quad \partial_t (\breve{x}^2 \Phi), \partial_t \breve{x} > 0.
    \]
\end{corollary}
\begin{proof}
    We already have the signs for $\partial_t (\breve{x}^2 \Phi), \partial_t \breve{\Phi}, \partial_t \breve{x}, \partial_t \breve{y}$ from \cref{breve_lemma,D_phi_derivative_sign}, and immediately infer the one for $\partial_t (\breve{y}^2 \Phi)$. Using~\eqref{bernoulli_free_nondim}, we get
    \[
        \partial_t \breve{D} = \tfrac{1}{2}\partial_t \breve{\Phi} - \tfrac{1}{2} \partial_t \parn*{\breve{x}^2 \Phi} < 0,
    \]
    which in turn implies $\partial_t \breve{B} < 0$, since~\eqref{solution_breve} yields that
    \[
        \breve{B}(t,R) = \frac{R}{R + t^3} \breve{D}(t,R)
    \]
    is a product of positive functions with negative derivative. Finally,
    \[
        \breve{H}(t,R) = \breve{y}(t,R) \breve{B}(t,R)
    \]
    gives us the sign of $\partial_t \breve{H}(t,R)$.
\end{proof}

As a result of \cref{monotonicity_corollary}, $\Omega_*$ may equally well be parameterized by \emph{any} combination of $R$ with another variable. In particular, $\Omega_*$ is the graph of a smooth function of $(H,R) \in (0,1)^2$, monotone in $H$, notably proving a claim from~\cite{dias_01_interfacial}. The price one has to pay is that such a parameterization can no longer be written down explicitly, since it requires solving a septic equation.

Furthermore, we can use \cref{monotonicity_corollary} to promptly obtain bounds on the nontrivial conjugate flows. They can only be improved if the bounds are allowed to depend on $R$. Note that the upper bound on $\mo{F}_1^2$ on $\Omega_{D>1}$ is \emph{slightly} larger than $1$.

\begin{corollary}
    The following bounds are sharp.
    \begin{gather*}
        1 < D < \tfrac{3}{2}, \quad \tfrac{4}{9} < H,B < 1, \quad 0 < F^2 < 1, \quad 0 < \mo{F}_1^2 < \tfrac{1}{4}(11 \sqrt{33} - 59), \quad 0 < \mo{F}_2^2 < \tfrac{1}{2}, \\
        0 < x < 1, \quad 1 < y < 2, \quad 0 < R < 1
        \shortintertext{on $\Omega_{D>1}$, while}
        0 < D < 1, \quad 0 < H,B < \tfrac{1}{2}, \quad 0 < F^2 < \tfrac{1}{2}, \quad 0 < \mo{F}_1^2 < \tfrac{1}{2}, \quad 0 < \mo{F}_2^2 < 2,                               \\
        1 < x < \infty, \quad 0 < y < 1, \quad 0 < R < 1
    \end{gather*}
    on $\Omega_{D<1}$.
\end{corollary}

\section{Uniqueness of the special family}
\label{sec uniqueness}

We end by proving the uniqueness of the solutions from \cref{special_section}.

\begin{theorem}
    \label{special_uniqueness}
    The choice $(D,F^2,x) = (3/7,1/7,3)$ from \cref{special_section} is the only selection of two or more constant variables from $D,H,B,R,F^2,x,y$ that leads to a family of solutions on $\Omega$.
\end{theorem}
\begin{proof}
    From~\eqref{parameterization_hat}, we deduce that no such pair can involve $y$, simply based on the non-matching degrees of the numerators and denominators with respect to $x$. The same is true for $x$ combined with either $B$ or $R$, leaving $D$, $H$, and $F^2$. The pair $(H,x)$ can be excluded because $y = x$ is always a root of the denominator of $\hat{H}$ in~\eqref{DHPhi_hat}, but never of the numerator. Meanwhile,~\eqref{D_hat_Phi_hat_derivatives} shows that $x = 3$ is the only choice leading to constant $D$ or $\Phi$, leading to the triple in the hypothesis.

    Not only are the pairs $(R,x)$ and $(R,y)$ excluded, we can use \cref{monotonicity_corollary} to exclude \emph{any} pair involving $R$. Furthermore, the pair $(H,B)$ is excluded because $y$ would necessarily be constant by~\eqref{conservation_of_mass_nondim_y}, and the pair $(D,F^2)$ leads to the existing triple by~\eqref{bernoulli_free_nondim}. This leaves only the potential candidates $(D,H)$, $(D,B)$, $(H,F^2)$, and $(B,F^2)$, for which we need to work slightly harder. We will use the following idea.

    Suppose that $f_1 = p_1/q_1$ and $f_2 = p_2/q_2$ are rational functions of $(x,y)$, and that $\alpha,\beta \in \R$ are such that
    \[
        f_1 = \alpha, \qquad f_2 = \beta,
    \]
    along some non-vertical curve $\tau \mapsto (x(\tau),y(\tau))$. In particular
    \[
        \res_y(p_1 - \alpha q_1, p_2 - \beta q_2)(x(\tau)) = 0
    \]
    for all $\tau$, whence
    \[
        \res_y(p_1 - \alpha q_1, p_2 - \beta q_2) \equiv 0
    \]
    by non-verticality. This is equivalent to a polynomial system of equations in $(\alpha,\beta)$.

    For the pair $(D,H)$, the resultant, using~\eqref{parameterization_hat}, reads
    \[
        (x - 3)(x - 1)^2(x + 1)\brak*{(3-\alpha)(\beta-1) +(\ldots)x +(1+\alpha-\alpha^2-3\beta+\alpha\beta)x^2 +(\ldots)x^3 +\alpha^2 x^4},
    \]
    where coefficients for $x^0$ and $x^4$ in the bracket force $(\alpha,\beta) = (0,1)$. Not only is this not physical, it also leads to a nonzero coefficient for $x^2$, and therefore no solution.

    For $(D,B)$, the resultant reads
    \[
        -(x - 3)(x - 1)^2(x + 1)\brak*{(\ldots)+ (2-2\alpha-\alpha^2+2\alpha \beta)x + (\ldots)x^2 + \alpha(\alpha-2\beta)x^3},
    \]
    where the only potentially physical choice for $x^3$ is $\alpha = 2\beta$. The coefficient for $x^1$ then reduces to $2(1 - 2\beta)$, so $(\alpha,\beta) = (1,1/2)$ is the only possibility. The other coefficients do in fact also vanish, however this choice cannot lead to solutions on $\Omega$, due to \cref{trivial_lemma}. In fact, the system $p_1 - q_1 = 0$ and $p_2 - \tfrac{1}{2}q_2 = 0$ has solution set $\hat{\Gamma}_{R=1}$ in the first quadrant.

    Moving on to $(H,F^2)$, we find
    \[
        (x - 3)(x - 1)^2 \brak*{(\ldots) +(4-\beta(\ldots))x + (\ldots)x^2 +(\ldots)x^3+\beta^2 x^5},
    \]
    where the coefficient for $x^5$ necessitates $\beta = 0$. In addition to being stationary, the coefficient in front of $x^1$ does not vanish. Hence, there are no solutions.

    Finally, we get
    \[
        -(x - 3)(x - 1)^2 \brak*{2(2-4\alpha+\beta+\alpha \beta) + \beta(\ldots) x + \beta (\ldots)x^2 + 3\beta^2 x^3 + \beta^2 x^4}
    \]
    for $(B,F^2)$, and therefore $(\alpha,\beta) = (1/2,0)$, which does not lead to solutions on $\Omega$. Indeed, $p_1 - \tfrac{1}{2} q_1 = 0$ and $q_1 = 0$ yield the same solutions as for $(D,B)$.
\end{proof}
\begin{remark}
    The result remains true even with the additional variables $\mo{F}_1^2$, $\mo{F}_2^2$ and $D - B$; with $\mo{F}_2^2 = 9/7$ on the same family. For the sake of simplicity, we chose to establish only the weaker version here.
\end{remark}

\section*{Acknowledgements}

Annie Grace Wang was supported by the Summer Undergraduate Research Award (SURA) from the University of Pittsburgh and by the National Science Foundation under Grant No.\ 2015553 from the AMS.

\printbibliography[title=References]

\end{document}